\documentclass[12pt]{article}

\usepackage[margin=0.950in]{geometry}

\usepackage{amsmath,amssymb,amsthm, csquotes}
\usepackage{newpxtext,newpxmath, tikz, cancel}
\usetikzlibrary{matrix, arrows.meta, positioning}
\usepackage{microtype}
\usepackage[authoryear,round]{natbib}
\usepackage{xcolor}
\newcommand{\N}{\mathbb N}
\newcommand{\Nzero}{\mathbb N_0}
\newcommand{\Q}{\mathbb Q}
\newcommand{\R}{\mathbb R}
\newcommand{\con}{\mathrm{con}}
\newcommand{\lex}{\mathrm{lex}}
\theoremstyle{definition}
\newtheorem{definition}{Definition}
\newtheorem{remark}{Remark}
\theoremstyle{plain}
\newtheorem{lemma}{Lemma}
\newtheorem{claim}{Claim}
\newtheorem{proposition}{Proposition}

\title{The existence of anonymous, stationary, monotone social welfare functions satisfying restricted dominance: an answer to an open question in intergenerational equity%
\footnote{We sincerely thank the Editor, the Associate Editor, Vicki Knoblauch, and the two anonymous referees for several useful comments and suggestions that improved the exposition of the paper. The first author would also like to thank Kuntal Banerjee for a discussion on this topic.}}
\author{Ram Sewak Dubey\thanks{%
Economics Department, Feliciano School of Business, Montclair State University, Montclair, NJ 07043, USA; E-mail: dubeyr@montclair.edu} \and 
Giorgio Laguzzi\thanks{%
Università del Piemonte Orientale, 15121 Alessandria, Italy; Email: giorgio.laguzzi@uniupo.it}  \and 
Francesco Ruscitti\thanks{%
Department of Economics, John Cabot University, Via della Lungara 233, 00165 Rome, Italy; Email: fruscitti@johncabot.edu}}
\date{\today}

\begin{document}
\maketitle

\begin{abstract}
In this note, we consider social welfare functions on infinite utility streams and answer an open question posed in \citet[Remark 1]{AsheimBanerjeeMitra2021}.
There exists a real-valued representation of social welfare preferences that satisfies monotonicity, stationarity, restricted dominance, and strong anonymity, hence also finite anonymity.
Our social welfare function satisfies a stronger finite-rank sensitivity condition, Liminf-Restricted Dominance, in the sense of \citet{asheim2022} and \citet{AsheimKamagaZuber2025}.
Our construction shows that limited discontinuity is not redundant in \citet[Proposition 2]{AsheimBanerjeeMitra2021}: once that diagonal continuity condition is dropped, even arbitrary-permutation invariance can coexist with stationarity and non-trivial sensitivity.

\noindent \emph{Keywords:} \texttt{Strong Anonymity,}\; \texttt{Restricted Dominance,}\; \texttt{Social Welfare Function,}\; \texttt{Stationarity.} 

\noindent \emph{Journal of Economic Literature} Classification Numbers: 
\texttt{D60,}\; \texttt{D70,}\; \texttt{D90.}
\end{abstract}
\section{Introduction}\label{sec:1}

In this short note we investigate an open question on the non-existence of social welfare function (SWF, henceforth) stated in \citet[Remark 1]{AsheimBanerjeeMitra2021}.
The authors establish impossibility of representation of social welfare preferences by a real-valued social welfare function when the preferences satisfy monotonicity, stationarity, restricted dominance, limited discontinuity and finite anonymity in Proposition 2 of the paper.
After proving the result, they explore if the impossibility result persists when the limited discontinuity condition is dropped and frame the conjecture as follows:

\medskip
\enquote{We are unaware of whether dropping Limited discontinuity allows for the existence of a monotone swf satisfying the remaining axioms of Proposition 2: Restricted dominance, Stationarity and Finite anonymity. 
We conjecture that there exists no such swf even if Limited discontinuity is dropped, but we have neither a proof of impossibility nor an example of a swf that satisfies the remaining axioms.
This issue remains an open question.}
\medskip

We present a monotone SWF  that satisfies restricted dominance, stationarity, and finite anonymity as a counterexample.
This SWF shows that the conjecture is false.
The social welfare function we construct satisfies stronger properties (than the ones in the conjecture) in two respects
First, it is invariant under all permutations of $\N$, not only finite permutations. 
Thus it satisfies Strong Anonymity, the impartiality axiom requiring evaluation to be invariant when arbitrary permutations are applied to an infinite utility stream. 
Second, the example satisfies a stronger form of restricted dominance, namely Liminf-Restricted Dominance: an improvement in a component increases social welfare whenever the component improved upon lies below the limit inferior of the resulting stream.
This is the sensitivity boundary emphasized by \citet{asheim2022} and used in the characterization of Strongly Anonymous Utilitarianism in \citet{AsheimKamagaZuber2025}. 
Our construction is not a strongly anonymous utilitarian criterion in their sense; rather, it is a lexicographic threshold-count criterion. 
Its role is to show that the impossibility result in \citet{AsheimBanerjeeMitra2021} comes from the continuity-type regularity assumption, not from stationarity and anonymity alone.

The rest of the paper is organized as follows.
In section 2 we collect the necessary definitions for examining the conjecture.
In section 3, we present the counterexample and show that it violates limited discontinuity as it must in view of \citet[ Proposition 2]{AsheimBanerjeeMitra2021}.
Section 4 contains brief concluding remarks.

\section{Framework and definitions}
\label{sec:2}
Let $\mathbb{R}$, $\mathbb{Q}$, and $\mathbb{N}$ be the set of real numbers, rational numbers, and natural numbers, respectively. 
We denote $\mathbb{N}\cup\{0\}$ by $\Nzero$ and $\Nzero\cup \{\infty\}$ by $A$.

Let $Y=[0,1]$ be the set of all possible utilities that any generation can achieve. 
Then, $X\equiv Y^{\mathbb{N}}$ is the set of all feasible utility streams. 
We denote an element of $X$ by $x =\left(x_{1}, x_{2},\cdots \right)$. 
For all $n\in \mathbb{N}$, $x_{n}\in Y$ represents the amount of utility that period-$n$ generation owns.
For $x,y\in X$, we write $x\ge y$ if $x_t\ge y_t$ for every $t\in\N$, and $x>y$ if $x\ge y$ and $x\ne y$.
If $b\in Y$, we write $b_{\con}=(b,b,b,\cdots)$.

A binary relation on $X$ is denoted by $\succsim$, with its symmetric and asymmetric parts, denoted by $\sim$ and $\succ$, respectively, defined in the usual manner. 
A \emph{social welfare order} (SWO) is, by definition, a complete and transitive binary relation on $X$. 
Given a SWO $\succsim$, we say that $\succsim$ can be \emph{represented} by a real-valued function, called a \emph{social welfare function} (SWF), if there is a mapping $W:X\rightarrow \mathbb{R}$ such that for all $x$, $y\in X$, $x\succsim y$ if and only if $W(x)\geq W(y)$.
We will be dealing with the following axioms that a SWF should satisfy

\begin{definition}[Monotonicity]
The SWF $W$ is monotone if, for all $x,y\in X$,
\[
  x>y \quad\Longrightarrow\quad W(x)\ge W(y).
\]
\end{definition}
\noindent A permutation $\pi$ of $\N$ is a bijection on $\N$.
Permutation $\pi$ of $\N$ is finite if there exists $\tau\in\N$ such that $\pi(t)=t$ for all $t\ge \tau$.
For $x\in X$, we let
\[
  x_{\pi}=(x_{\pi(1)},x_{\pi(2)},\ldots).
\]

\begin{definition}[Finite anonymity]

The SWF $W$ satisfies finite anonymity if, for every $x\in X$ and every finite permutation $\pi$ of $\N$,
\[
  W(x_{\pi})=W(x).
\]
\end{definition}

\begin{definition}[Strong anonymity]
The SWF $W$ satisfies strong anonymity if, for every $x\in X$ and every permutation $\pi$ of $\N$,
\[
  W(x_{\pi})=W(x).
\]
\end{definition}

\begin{definition}[Restricted dominance]
The SWF $W$ satisfies restricted dominance if, for all $a,b\in Y$ with $a<b$,
\[
  W(a,b_{\con})<W(b_{\con}).
\]
\end{definition}
\noindent For any $x\in X$, write
\[
  \ell(x)=\liminf_{t\to\infty}x_t .
\]

\begin{definition}[Liminf-Restricted Dominance]
The SWF $W$ satisfies Liminf-Restricted Dominance if, for all $x,y\in X$ such that $x_i>y_i$ for some $i\in\N$ and $x_j=y_j$ for all $j\in \N\setminus\{i\}$,
\[
  \ell(x)>y_i
  \quad\Longrightarrow\quad
  W(x)>W(y).
\]
\end{definition}

\begin{definition}[Stationarity]
The SWF $W$ satisfies stationarity if, for all $x,y\in X$ with $x_1=y_1$,
\[
  W(x)\ge W(y)
  \quad\Longleftrightarrow\quad
  W(x^{\prime})\ge W(y^{\prime}),
\]
where $x^{\prime}= (x_2, x_3, \cdots)$ and $y^{\prime} = (y_2, y_3, \cdots)$.
\end{definition}
\noindent For comparison with the impossibility result  of \citet[ Proposition 2]{AsheimBanerjeeMitra2021}, we define the diagonal mapping
\[
 f_W:[0,1]\to\R,
  \qquad
  f_W(c)=W(c_{\con}).
\]
\begin{definition}[Limited discontinuity]
The SWF $W$ satisfies limited discontinuity if the set of discontinuity points of $f_W$ is not dense in $[0,1]$.
\end{definition}
\section{A counterexample}
In this section, we construct an example of a monotone SWF satisfying restricted dominance, stationarity, and finite anonymity that proves the conjecture to be false.
It is useful to recall the definition of lexicographic order.
\begin{definition} (Lexicographic order $\geq_{\lex}$ on  $A^{\N}$) 
For all $x, y\in A^{\N}$, $x\geq_{\lex}y$ if and only if either $x_{t}=y_{t}$ for all $t\in \N$, or $x_{T}>y_{T}$ for some $T\in \N$ and $x_{t}=y_{t}$ for all $t\in \N$ such that $t<T$.
\end{definition}
We first describe a lexicographic order on $A^{\N}$ and show that it admits a real-valued representation.
\begin{lemma}\label{L1}
\emph{There exists a real-valued function $\Phi: A^{\N}\to\R$ on $A^{\N}$ such that for all $u, v\in A^{\N}$,}
\[
  u\le_{\lex}v
  \quad\Longleftrightarrow\quad
  \Phi(u)\le \Phi(v).\]
\end{lemma}

\begin{proof}
We construct the function explicitly.
We define intervals $J_a$, one for each $a\in A$ as follows.
(i) For $n\in\Nzero$, we set
\[
  J_n=\left(\frac{1}{2}-\frac{1}{2^{2n+2}},\frac{1}{2}-\frac{1}{2^{2n+3}}\right).%
 \footnote{Thus $J_0=(\frac{1}{4},\frac{3}{8})$, $J_1=(\frac{7}{16},\frac{15}{32})$, $\cdots$.}
\]
(ii) For $a=\infty$, we set
\[
  J_{\infty}=\left(\frac{3}{4},\frac{7}{8}\right).
\]
Observe that (i) the length of  interval $J_n$ is $\frac{1}{2^{2n+3}}$; (ii) for $i\ne j$, the closures $\bar{J}_i$ and $\bar{J}_j$ are disjoint, i.e., $\bar{J}_i\cap\bar{J}_j =\emptyset$, and  (iii) ordered i.e., for $i<j$, the interval $J_j$ lies entirely to the right of $J_i$. 

\noindent Let $A^{<\N}$ denote the set of finite sequences with each element belonging to $A$.
We take $I_{\emptyset}=(0,1)$, and define recursively a sequence of sets by successive concatenation, starting from $I_{\emptyset}$ and then moving successively to $I_{\emptyset^{\frown}n_1}$, followed by $I_{\emptyset^{\frown}n_1^{\frown}n_2}$, etc. in the following manner.
Suppose 
\[
I_{\sigma}=I_{\emptyset^{\frown}\sigma_{1}^{\frown}\cdots^{\frown}\sigma_{n}} =(\ell_{\sigma}, r_{\sigma})
\]
has been defined for every finite sequence $\sigma=(\sigma_{1},\cdots,\sigma_{n})\in A^{n}$ of length $n$.
For $a\in A$, we write $J_a=(\alpha_a,\beta_a)$ and define
\[
  I_{\sigma^\frown a}
  =
  \bigl(\ell_{\sigma}+(r_{\sigma}-\ell_{\sigma})\alpha_a,
        \ell_{\sigma}+(r_{\sigma}-\ell_{\sigma})\beta_a\bigr).
\]
The intervals $I_{\sigma^\frown a}$ are contained in $I_{\sigma}$, have ordered disjoint closures, and have length at most $\frac{1}{8}$ times the length of $I_{\sigma}$.%
\footnote{Consider for example, $I_{\emptyset^\frown 0} = \left(0+(1-0)\cdot \frac{1}{4}, 0+(1-0)\cdot \frac{3}{8}\right) =  \left(\frac{1}{4}, \frac{3}{8}\right)$, and $I_{\emptyset^\frown 1} = \left(0+(1-0)\cdot \frac{7}{16}, 0+(1-0)\cdot \frac{15}{32}\right) =  \left(\frac{7}{16}, \frac{15}{32}\right)$. 
It is easy to verify that $I_{\emptyset^\frown 0 ^\frown1} =   \left(\frac{39}{128}, \frac{79}{256}\right)$.}
For $u=(u_1,u_2,\cdots)\in A^{\N}$, let
\[
  u|m=(u_1,\ldots,u_m)
\]
be its initial segment of length $m$.
The closed intervals
\[
  \overline{I_{u|1}}
  \supseteq
  \overline{I_{u|2}}
  \supseteq
  \overline{I_{u|3}}
  \supseteq\cdots
\]
are nested, and their lengths tend to zero.
Hence their intersection is a singleton.
Define $\Phi(u)$ to be the unique number in that intersection.

Now let $u<_{\lex}v$.
Let $k$ be the first coordinate at which $u$ and $v$ differ.
Then $u_i=v_i$ for $i<k$ and $u_k<v_k$.
Let
\[
  \sigma=(u_1,\ldots,u_{k-1})=(v_1,\ldots,v_{k-1}).
\]
By construction, $\overline{I_{\sigma^\frown u_k}}$ lies strictly to the left of $\overline{I_{\sigma^\frown v_k}}$.
Since
\[
  \Phi(u)\in\overline{I_{\sigma^\frown u_k}},
  \qquad
  \Phi(v)\in\overline{I_{\sigma^\frown v_k}},
\]
we get $\Phi(u)<\Phi(v)$.
Thus strict lexicographic order is preserved.
Further, if $\Phi(u)\le \Phi(v)$ and $\neg u\le_{\lex}v$, then $v<_{\lex}u$ since the lexicographic relation is complete, hence $\Phi(v)<\Phi(u)$, a contradiction.
Therefore $u\le_{\lex}v$ if and only if $\Phi(u)\le\Phi(v)$.
\end{proof}

We take $(q_k)_{k\in\mathbb N}$ to be the one-to-one enumeration of $\mathbb Q\cap(0,1)$ as detailed in the appendix.
We call this enumeration $\Q(01)$, i.e., $\Q(01)= (q_1, q_2, \cdots)$.
For $x\in X=Y^{\N}$ and $k\in\N$, define
\[
N_k(x)=|\{t\in\N:x_t<q_k\}|\in A,
\]
where  $|\cdot|$ denote the cardinality of a given finite (or infinite) set.
We refer to $N_k(x)$ as cut-count.
We denote the profile
\[
P(x)=(N_1(x), N_2(x),\ldots)\in A^{\N}.
\]
Profiles are compared lexicographically, with smaller cut-counts being socially preferred.
We define the SWF by using the function $\Phi$ described in Lemma \ref{L1} as follows.
\[
 W(x)= -\Phi(P(x)) = -\Phi\bigl(N_1(x),N_2(x),\cdots\bigr).%
 \footnote{For instance, let $x=(0_{\con})$, and $y=(1_{\con})$. 
It is not difficult to verify that $W(x)=-\frac{6}{7}$ and $W(y) = -\frac{2}{7}$.}
\]
The minus sign reflects the convention that smaller cut-count profiles represent higher welfare.

\begin{proposition}\label{P1}
\emph{Let $Y=[0,1]$ and $X=Y^{\N}$.
The function
\[
W(x)=-\Phi\bigl(N_1(x),N_2(x),\cdots\bigr)
\]
is a real-valued monotone SWF satisfying strong anonymity, Liminf-Restricted Dominance, restricted dominance, and stationarity.}
\end{proposition}

\begin{proof}
By Lemma \ref{L1}, $\Phi$ is a real-valued function on  $A^{\N}$.
Since $P(x)\in A^{\N}$ for every infinite utility stream $x$, the value $W(x)=-\Phi(P(x))$ is well-defined.

\noindent\textbf{Monotonicity.}
Suppose $x, y\in X$ and $x_t\ge y_t$ for all $t\in \N$.
For every $q_k$,
\[
\{t:x_t<q_k\}\subseteq \{t:y_t<q_k\}.
\]
Hence
\[
  N_k(x)\le N_k(y)
  \qquad\text{for every }k.
\]
Therefore, $P(x) \le_{\lex} P(y)$. 
Hence,  by Lemma \ref{L1} and by the above definition of $W$, $W(x)\geq W(y)$.
Thus the SWF $W(\cdot)$ is monotone.

\noindent\textbf{Strong anonymity.}
Let $\pi$ be an arbitrary permutation of $\N$.
For every $k$,
\[
  |\{t:x_{\pi(t)}<q_k\}|
  =
  |\{t:x_t<q_k\}|.
\]
Thus $N_k(x_{\pi})=N_k(x)$ for every $k$.
Hence $P(x_{\pi})=P(x)$ and $W(x_{\pi})=W(x)$.
Therefore strong anonymity holds, and finite anonymity follows immediately.

\noindent\textbf{Liminf-Restricted Dominance.}
Let $x,y\in X$ be such that $x_i>y_i$ for some $i\in\N$ and $x_j=y_j$ for all $j\ne i$.
Assume that $\ell(x)>y_i$.
Since $x\ge y$, for every rational threshold $q_k$,
\[
  N_k(x)\le N_k(y).
\]
Choose a rational $q_m\in\Q(01)$ such that
\[
  y_i<q_m<\min\{x_i,\ell(x)\}.
\]
This is possible by density of the rationals.
Since $q_m<\ell(x)$, the set $\{t:x_t<q_m\}$ is finite.
Moreover, because $x_i>q_m>y_i$ and $x_j=y_j$ for all $j\ne i$,
\[
  N_m(y)=N_m(x)+1.
\]
Thus the profiles $P(x)$ and $P(y)$ differ at some coordinate, and at every coordinate at which they differ the count for $x$ is weakly smaller, with strict inequality at coordinate $m$.
Therefore the first coordinate at which $P(x)$ and $P(y)$ differ favors $x$:
\[
  P(x)<_{\lex}P(y).
\]
Since $W=-\Phi(P(\cdot))$, it follows that
\[
  W(x)>W(y).
\]
Hence $W$ satisfies Liminf-Restricted Dominance.
Since Liminf-Restricted Dominance implies restricted dominance (by letting $x=b_{\con}$, $y=\left(a, b_{\con}\right)$, with $a<b$, $i=1$, $x_{i}=b$ and $y_{i}=a$), it follows that $W$ also satisfies restricted dominance.

The following claim would be useful for showing that the SWF satisfies stationarity.

\begin{claim}
Let $u,v\in A^{\mathbb N}$ and let $e\in\{0,1\}^{\mathbb N}$.
Define coordinate-wise addition, with $\infty+0=\infty+1=\infty$.
Then
\[
u\le_{\lex} v
\quad\Longleftrightarrow\quad
u+e\le_{\lex} v+e.
\]
\end{claim}

\begin{proof}
$\Longrightarrow$ Suppose that $u\leq_{\lex}v$. 
We distinguish between two cases: if $u=v$, the claim is immediate. 
For the second case, let $u_{k}<v_{k}$ and $u_{j}=v_{j}$ for all $j<k$. 
Thus, $u+e\leq_{\lex}v+e$.

$\Longleftarrow$ Suppose that $u+e\leq _{\lex}v+e$.
Case 1: $u+e=v+e$ which implies $u=v$ and 
$u\leq _{\lex}v$. 
Case 2: $u_{k}+e_{k}<v_{k}+e_{k}$ and $u_{j}+e_{j}=v_{j}+e_{j}$ for all $j<k$. 
Then, $u_{k}<v_{k}$ and $u_{j}=v_{j}$ for all $j<k$. 
Therefore, $u\leq _{\lex}v$.
\end{proof}

\noindent\textbf{Stationarity.}
Let $x,y\in X$ satisfy $x_1=y_1=a$.
For each $k$,
\[
  N_k(x)=N_k(x^{\prime})+\mathbf 1_{\{a<q_k\}}, \;\text{and}\;
  N_k(y)=N_k(y^{\prime})+\mathbf 1_{\{a<q_k\}},%
  \footnote{Recall the definition
 \begin{equation*}
 \mathbf 1_{\{a<q_k\}} =
\begin{cases}
1 & \text{if}\; a<q_{k},\\
0 & \text{if}\;  a\geq q_{k}.
\end{cases}
\end{equation*}}
\]
when the cut-counts are all finite.
Equality holds trivially in case the cut-counts are infinite as we define $\infty+1=\infty$ and $\infty+0=\infty$.
Observe that the vector
\[
e(a)=\bigl(\mathbf 1_{\{a<q_1\}},\mathbf 1_{\{a<q_2\}},\cdots\bigr)
\]
is added coordinate-wise to both profiles $P(x^{\prime})$ and $P(y^{\prime})$ to obtain  the profiles $P(x)$ and $P(y)$ respectively.
Therefore, lexicographic order on $A^{\N}$ is preserved.
Therefore
\[
  P(x)\le_{\lex}P(y)
  \quad\Longleftrightarrow\quad
  P(x^{\prime})\le_{\lex}P(y^{\prime}).
\]
Since $W=-\Phi(P(\cdot))$, this is equivalent to
\[
  W(x)\ge W(y)
  \quad\Longleftrightarrow\quad
  W(x')\ge W(y').
\]
Thus stationarity holds.
\end{proof}

\begin{remark}\label{R1}
Proposition \ref{P1} proves that the conjecture in \citet[Remark 1]{AsheimBanerjeeMitra2021} is false.
Further, in view of \citet[Proposition 2]{AsheimBanerjeeMitra2021}, the SWF must violate the limited discontinuity axiom. 
We show that is indeed the case.
Let
\[
  f_W(c)=W(c_{\con}).
\]
We show that every $q_k\in \Q(01)$ is a discontinuity point of $f_W$.
Fix $q_m$. 
Choose $\varepsilon>0$ such that $\varepsilon<q_m$ and the interval $(q_m-\varepsilon,q_m+\varepsilon)$ contains no member of the finite set $\{q_1,\ldots,q_{m-1}\}$.
For every $b\in(q_m-\varepsilon, q_m)$, the first $m-1$ coordinates of $P(b_{\con})$ and $P((q_m)_{\con})$ agree, while at coordinate $m$,
\[
  N_m(b_{\con})=\infty,
  \qquad
  N_m((q_m)_{\con})=0.
\]
Let $\sigma=(N_1((q_m)_{\con}),\cdots, N_{m-1}((q_m)_{\con}))$.
Then \(\Phi(P((q_m)_{\con}))\) lies in the closed interval \(\overline{I_{\sigma^\frown 0}}\), while every such \(\Phi(P(b_{\con}))\) lies in \(\overline{I_{\sigma^\frown \infty}}\).
By the construction in Lemma \ref{L1}, these two closed intervals are separated by a positive gap.
Therefore $\Phi(P((q_m)_{\con}))$ is separated by a positive distance from all $\Phi(P(b_{\con}))$ with $b\in(q_m-\varepsilon,q_m)$.
Since $W=-\Phi(P(\cdot))$, $f_W$ is not left-continuous at $q_m$.
Hence $q_m$ is a discontinuity point of $f_W$.
Because \(\mathbb Q\cap(0,1)\) is dense in \([0,1]\), the discontinuity set of \(f_W\) is dense in \([0,1]\).
So, $W$ does not satisfy limited discontinuity.
\end{remark}

\section{Conclusion}\label{sec:4}

Proposition \ref{P1} answers an open question posed in \citet[Remark 1]{AsheimBanerjeeMitra2021}. 
Moreover, the SWF constructed in the proposition satisfies Strong Anonymity and Liminf-Restricted Dominance.
The SWF employs a lexicographic threshold-count criterion, and it violates the continuity conditions used in \citet{AsheimBanerjeeMitra2021}, \citet{asheim2022}, and \citet{AsheimKamagaZuber2025}. 
Thus, the example does not conflict with those results; it shows that, without such continuity (a regularity condition), strong anonymity, stationarity, real-valued representability, and Liminf-Restricted Dominance can coexist.

\appendix
\section*{Appendix: Enumeration of rational numbers in $(0,1)$ for use in Proposition 1}
{\footnotesize
\[
\begin{tikzpicture}[
  every node/.style={font=\large},
  arr/.style={-{Latex[length=2mm]}, thick}
]
\matrix (m) [matrix of math nodes,
  row sep=0.50cm,
  column sep=0.75cm
]{\scriptstyle
\dfrac12 & \dfrac13 & \dfrac14 & \dfrac15 & \dfrac16 & \cdots\\
\dfrac23 & \cancel{\dfrac24} & \dfrac25 & \cancel{\dfrac26} & \dfrac27 & \cdots\\
\dfrac34 & \dfrac35 & \cancel{\dfrac36} & \dfrac37 & \dfrac38 & \cdots\\
\dfrac45 & \cancel{\dfrac46} & \dfrac47 & \cancel{\dfrac48} & \dfrac49 & \cdots\\
\dfrac56 & \dfrac57 & \dfrac58 & \dfrac59 & \cancel{\dfrac5{10}} & \cdots\\
\vdots   & \vdots   & \vdots   & \vdots   & \vdots      & \ddots\\
};
\draw[arr] (m-1-1) -- (m-1-2);
\draw[arr] (m-1-2) -- (m-2-1);
\draw[arr] (m-2-1) -- (m-3-1);
\draw[arr] (m-3-1) -- (m-2-2);
\draw[arr] (m-2-2) -- (m-1-3);
\draw[arr] (m-1-3) -- (m-1-4);
\draw[arr] (m-1-4) -- (m-2-3);
\draw[arr] (m-2-3) -- (m-3-2);
\draw[arr] (m-3-2) -- (m-4-1);
\draw[arr] (m-4-1) -- (m-5-1);
\draw[arr] (m-5-1) -- (m-4-2);
\draw[arr] (m-4-2) -- (m-3-3);
\draw[arr] (m-3-3) -- (m-2-4);
\draw[arr] (m-2-4) -- (m-1-5);
\node[below=0.9cm of m-6-3, font=\normalsize]
{Fixed enumeration of rational numbers in  $\mathbb{Q}(01)$.};
\end{tikzpicture}
\]}
The crossed-out fractions are omitted because they repeat rational numbers that have already appeared in the enumeration.
The first ten rational numbers are:
\[
q_1=\frac12,\,
q_2=\frac13,\,
q_3=\frac23,\,
q_4=\frac34,\,
q_5=\frac14,\,
q_6=\frac15,\,
q_7=\frac25,\,
q_8=\frac35,\,
q_9=\frac45,\,
q_{10}=\frac56,\,
\cdots
\]

\end{document}